\documentclass[12pt]{article}
\usepackage[intlimits]{amsmath}
\usepackage{amssymb,amsthm,slashed,mathabx}

\usepackage[T1]{fontenc}
\usepackage[cp1250]{inputenc}

\usepackage[showonlyrefs]{mathtools}
\mathtoolsset{showonlyrefs=true}
\usepackage{fixltx2e}
\MakeRobust{\eqref}

\newcommand{\<}{\langle}
\renewcommand{\>}{\rangle}
\newcommand{\D}{\mathcal{D}}
\newcommand{\Hc}{\mathcal{H}}
\newcommand{\mR}{\mathbb{R}}
\newcommand{\mC}{\mathbb{C}}
\newcommand{\mZ}{\mathbb{Z}}
\newcommand{\1}{{\bf 1}}
\newcommand{\al}{\alpha}
\newcommand{\ga}{\gamma}
\newcommand{\la}{\lambda}
\newcommand{\dsp}{\displaystyle}
\newcommand{\vep}{\varepsilon}
\newcommand{\vph}{\varphi}
\newcommand{\w}{\omega}
\newcommand{\ba}{\bar{a}}
\newcommand{\bb}{\bar{b}}
\newcommand{\W}{\Omega}
\newcommand{\ov}{\overline}
\newcommand{\p}{\partial}
\newcommand{\con}{\mathrm{const}}

\newtheorem{pr}{Proposition}

\DeclareMathOperator{\Ker}{Ker}
\DeclareMathOperator{\modulo}{mod}

\title{Generalized uncertainty relations}
\author{Andrzej Herdegen\thanks{e-mail: andrzej.herdegen@uj.edu.pl} {} and Piotr Ziobro\thanks{e-mail: piotr.ziobro@uj.edu.pl}\\
{\it Institute of Physics, Jagiellonian University,}\\
{\it ul.\,S.\,{\L}ojasiewicza 11, 30-348  Krak\'{o}w, Poland}}
\date{}

\begin{document}

\maketitle

\begin{abstract}
 The standard uncertainty relations (UR) in quantum mechanics are typically used for unbounded operators (like the canonical pair). This implies the need for the control of the domain problems. On the other hand, the use of (possibly bounded) functions of basic observables usually leads to more complex and less readily interpretable relations. Also, UR may turn trivial for certain states if the commutator of observables is not proportional to a positive operator. In this letter we consider a generalization of standard UR resulting from the use of two, instead of one, vector states. The possibility to link these states to each other in various ways adds additional flexibility to UR, which may compensate some of the above mentioned drawbacks. We discuss applications of the general scheme, leading not only to technical improvements, but also to interesting new insight.

\vspace{1ex}

\end{abstract}

\section{Introduction}\label{stand}

In popular textbook terms, the quantum-mechanical UR states that if any three observables $A,B,C$ satisfy the commutation relation
\begin{equation}\label{commut}
 [A,B]=iC\,,
\end{equation}
then in any normalized quantum state $\psi$ there is
\begin{equation}\label{uncert}
 \Delta_\psi(A)\Delta_\psi(B)\geq \tfrac{1}{2}|\<C\>_\psi|\,,
\end{equation}
where $\<A\>_\psi=(\psi,A\psi)$ is the mean value of the probability distribution of $A$ in the state $\psi$ and $\Delta_\psi(A)=[(\psi,(A-\<A\>_\psi)^2\psi)]^{1/2}$ is the standard deviation of this distribution (see, e.g., \cite{gr95}). This formulation follows the extension of the original Heisenberg relation \cite{he27} given by Robertson \cite{ro29}.

Quantum-mechanical observables are self-adjoint operators, which usually are not bounded (as is the case with the most prominent example of the~canonical pair). Therefore, a more mathematically conscious formulation of UR has to take into account domain restrictions (see, e.g., \cite{ha13}). By $\D(A)\subseteq\Hc$ we denote the domain of an operator $A$ acting in a Hilbert space~$\Hc$. For any self-adjoint operators $A,B$ the relation \eqref{commut} defines a symmetric operator $C$ on the domain
\begin{equation}\label{domC}
 \D(C)=\D(AB)\cap\D(BA)\,.
\end{equation}
It is now a simple mathematical theorem, that the relation \eqref{uncert} is satisfied for any normalized vector $\psi\in\D(C)$.

The question that now arises is this: is $\D(C)$ ``sufficiently large'' for the relation to be of use? In the worst possible case it could happen that $\D(C)$ would not be dense in $\Hc$, which would leave outside the range of the relation the whole closed subspace $\D(C)^\bot\subseteq\Hc$. This case is of little use, so we assume from now on that $\ov{\D(C)}=\Hc$ (bar denoting the closure). Even with this restriction, we are still left with a few open problems:
\begin{itemize}
\item[(i)] If $\psi$ is not in $\D(A)$, then $\Delta_\psi(A)$ may be regarded as infinite; the relation tells us then nothing on the spread of distribution of $B$ in the state $\psi$.
\item[(ii)] If $\psi\in\D(A)\cap\D(B)$, then the product of uncertainties is finite. It may happen that also $C$ extends to this larger domain, but UR need not extend to this case.
\item[(iii)] The restrictions of $A$ and $B$ to $\D(C)$ need not determine self-adjoint operators $A$ and $B$ uniquely (i.e., in technical terms, $\D(C)$ need not be a core for $A$ and $B$), so that the UR does not admit all states crucial for the determination of the observables $A$ and $B$ themselves.
\item[(iv)] In general, if $C$ is not strictly positive, it may have vanishing expectation value $\<C\>$ in the state $\psi$ under consideration. The relation has no nontrivial content in this case.
\end{itemize}

Standard examples of the difficulties (ii) and (iii) occur for the `angle -- angular momentum' pair. Let $\Phi,L$ be the operators in $\Hc=L^2(\<0,2\pi\>)$ defined by:
\begin{gather*}
 \D(\Phi)=\Hc\,,\ \D(L)=\{\psi\in \Hc\mid \psi'\in\Hc,\ \psi(0)=\psi(2\pi)\}\label{L}\\
 (\Phi\psi)(\vph)=\vph\psi(\vph)\,,\quad (L\psi)(\vph)=-i\psi'(\vph)\,,
\end{gather*}
(where $\psi'$ is a measurable derivative function of $\psi$). With these domains both operators are self-adjoint. The UR \eqref{uncert} on the domain \eqref{domC} holds then with $C\psi=\psi$ on
\begin{equation}\label{C}
 \D(C)=\{\psi\in \Hc\mid \psi'\in\Hc,\ \psi(0)=\psi(2\pi)=0\}\,,
\end{equation}
so that on this domain, which may be shown to be dense: $\ov{\D(C)}=\Hc$, one has for normalized $\psi$:
\begin{equation}\label{PhiL}
 \Delta_\psi(\Phi)\Delta_\psi(L)\geq\tfrac{1}{2}\,.
\end{equation}
Both sides of this inequality are meaningful and finite for $\psi\in\D(\Phi)\cap\D(L)=\D(L)$, but the inequality does not extend to this larger domain: take any eigenstate of $L$ to find $0\geq 1$, which illustrates difficulty (ii). The explanation of this seeming paradox is that for \mbox{$L\psi=m\psi$} there is no sequence of vectors $\psi_n\in\D(C)$ which converges to $\psi$, and at the same time satisfies $\Delta_{\psi_n}(L)\rightarrow \Delta_\psi(L)$.

Even more disturbing is the fact that the domain $\D(C)$ is not sufficient to uniquely identify the self-adjoint operator $L$ taking part in the above relation (i.e., $\D(C)$ is not a core for $L$). To see this it is sufficient to note that for each complex $\w$ with $|\w|=1$ one has a self-adjoint operator $L_\w$ defined as $L$, but on a different domain $\D(L_\w)=\{\psi\in \Hc\mid \psi'\in\Hc,\ \psi(2\pi)=\w\,\psi(0)\}$ (see, e.g., \cite{rs72}).
Each of these operators may replace $L$ in the above UR, with no change of $C$ or $\D(C)$. This illustrates difficulty (iii).

Difficulty (iv) occurs in the well-known case of $3$-dimensional angular momentum operators; see below.

In the rest of this article we propose a simple extension of the minimization argument leading to UR. Our generalized UR are given in Section 2, Proposition 1. Applications of the general scheme to a few cases of particular physical interest are discussed in Section 3. Our use of the result of Proposition 1 is closely related to the mathematical and physical problem of the original formulation of UR: find optimal bounds on spreads of probability distributions defined by a given state for two incompatible observables.\footnote{A different generalization of UR has been recently proposed in \cite{ma14}.} We do not touch upon various other problems broadly related to `uncertainty' in quantum mechanics, which recently draw considerable attention in physical literature (entropic uncertainty, error-disturbance problem, parameter estimation etc.).

\section{Generalized uncertainty relation}

We consider general normal operators. We recall that (see, e.g., \cite{ru73}):
\begin{itemize}
\item[(i)] $A$ is called normal if it is densely defined, closed and satisfies $A^*A=AA^*$.
\item[(ii)] For normal $A$ there is $\D(A^*)=\D(A)$ and $\|A^*\chi\|=\|A\chi\|$ for all vectors \mbox{$\chi\in\D(A)$}.
\item[(iii)] All functions of self-adjoint operators are normal operators.
\item[(iv)] All normal operators satisfy the spectral theorem, $A=\int_{\sigma(A)}z\,dE^A_z$, where $\sigma(A)$ is the spectrum of $A$ -- a~closed subset of $\mC$, and $E^A_\W$ is its spectral family.  For each normalized vector $\psi$ the mapping $\W\mapsto (\psi,E^A_\W\psi)$ is a probability measure on Borel sets $\W\subseteq\sigma(A)$, with the mean value $\<A\>_\psi=(\psi,A\psi)\in\mC$ and the standard deviation $\Delta_\psi(A)=\|(A-\<A\>_\psi)\psi\|$.
\end{itemize}
Let $A$ and $B$ be normal operators with the domains $\D(A)$ and $\D(B)$, respectively.  Then we define a sesqui-linear form
\begin{equation*}\label{form}
 q_{A,B}(\vph,\chi)=(A^*\vph,B\chi)-(B^*\vph,A\chi)
\end{equation*}
with the domain $\vph,\chi \in\D(q_{A,B})\equiv \D(A)\cap\D(B)$ (an example of a recent use of this `weak commutator' may be found in \cite{ta11}). In special case when $\chi\in\D(AB)\cap\D(BA)$ this weak commutator becomes the ordinary one, $q_{A,B}(\vph,\chi)=(\vph,[A,B]\chi)$. Moreover, for any complex numbers $a$, $b$, we denote $A_a=A-a\1$, $B_b=B-b\1$, which are also normal operators.
\begin{pr}[Generalized Uncertainty Relation]\label{gur}\ \\
For any normal operators $A$, $B$ and unit vectors $\vph,\chi\in\D(q_{A,B})$ the following inequality holds
\begin{multline}\label{guncert}
 |q_{A,B}(\vph,\chi)|\leq \inf_{a,b\in\mC}\Big(\|(A_a\vph\|\|B_b\chi\|+\|B_b\vph\|\|A_a\chi\|\Big)\\
 =\inf_{\begin{smallmatrix}\la_1,\la_2\in\<0,1\>\\ \la_1+\la_2=1\end{smallmatrix}}
 \bigg\{\sqrt{\Delta_\vph^2(A) + |\delta\<A\>|^2\la_1^2}\,\sqrt{\Delta_\chi^2(B) + |\delta\<B\>|^2\la_2^2}\\
 +\sqrt{\Delta_\vph^2(B) + |\delta\<B\>|^2\la_1^2}\,\sqrt{\Delta_\chi^2(A) + |\delta\<A\>|^2\la_2^2}\bigg\}\,,
\end{multline}
where $\delta\<A\>=\<A\>_\vph-\<A\>_\chi$, $\delta\<B\>=\<B\>_\vph-\<B\>_\chi$.
\end{pr}
\begin{proof}
We first note that $q_{A_a,B_b}=q_{A,B}$. Therefore, the successive use of the triangle and the Schwarz inequalities (and property~(ii) above) gives
\begin{equation}\label{qAB}
 |q_{A,B}(\vph,\chi)|\leq\|A_a\vph\|\|B_b\chi\|+\|B_b\vph\|\|A_a\chi\|\,.
\end{equation}
Thus, using the arbitrariness of $a$ and $b$, we arrive at the first relation (inequality) in \eqref{guncert}. The rhs of \eqref{qAB} is a~real function $F(a,\ba, b, \bb)$, nondecreasing for $|a|$ or $|b|$ sufficiently large and tending to $+\infty$. Therefore, it reaches its infimum at one of its stationary points, which are the solutions of the set of equations $\p F/\p\ba=0$, $\p F/\p\bb=0$, i.e.
\begin{gather}
\ga_2\big(a-\<A\>_\vph\big)+\ga_1\big(a-\<A\>_\chi\big)=0\,,\\
\ga_2\big(b-\<B\>_\vph\big)+\ga_1\big(b-\<B\>_\chi\big)=0\,,\\
\ga_1=\|A_a\vph\|\,\|B_b\vph\|\,,\ \ga_2=\|A_a\chi\|\,\|B_b\chi\|\,.
\end{gather}
Solving the first two of these equations for $a$ and $b$ in terms of $\la_i\equiv \ga_i/(\ga_1+\ga_2)$ ($i=1,2$), one obtains $a=\la_2\<A\>_\vph+\la_1\<A\>_\chi$, $b=\la_2\<B\>_\vph+\la_1\<B\>_\chi$. Setting these values into \eqref{qAB} one obtains the second relation (equality) in \eqref{guncert}. The condition for stationary points is now reduced to the condition on $\la_i$: $\ga_1\la_2=\ga_2\la_1$, i.e.
\begin{multline}\label{stat}
\Big(\Delta_\vph^2(A) + |\delta\<A\>|^2\la_1^2\Big)\Big(\Delta_\vph^2(B) + |\delta\<B\>|^2\la_1^2\Big)\la_2^2\\
=\Big(\Delta_\chi^2(A) + |\delta\<A\>|^2\la_2^2\Big)\Big(\Delta_\chi^2(B) + |\delta\<B\>|^2\la_2^2\Big)\la_1^2\,,
\end{multline}
where $\la_1,\la_2\in\<0,1\>$, $\la_1+\la_2=1$. This condition, when expressed in terms of one unknown, is an algebraic equation of fifth order, which is not algebraically solvable in general. Thus we postpone its solution to more special cases.
\end{proof}

The extended flexibility of the relation \eqref{guncert} relies on the possibility to assume more general relations between $\vph$ and $\chi$, than equality. We shall discuss examples in the next section. But we end this section with this simple special case.

For $\vph=\chi$, $\|\chi\|=1$,  one finds easily that the infimum is reached for $a=\<A\>_\chi$ and $b=\<B\>_\chi$, so
\begin{equation}
  \tfrac{1}{2}|q_{A,B}(\chi,\chi)|\leq \Delta_\chi(A)\Delta_\chi(B)\,,\quad \chi\in\D(A)\cap\D(B)\,.
\end{equation}
The UR in this form was applied by Kraus \cite{kr65} to the above-mentioned case of angle -- angular momentum pair. Integrating by parts one finds that  $q_{L,\Phi}(\chi,\chi)=i(2\pi|\chi(2\pi)|^2-1)$, hence
\[
 \tfrac{1}{2}\big|1-2\pi |\chi(2\pi)|^2\big|\leq\Delta_\chi(L)\Delta_\chi(\Phi)\,,\quad \chi\in\D(L)\,.
 \]
This yields correct $0\leq0$ for $L\chi=m\chi$, but a drawback of this relation is that beside uncertainties it needs the value of $\chi$ at a particular point.

\section{Applications of the generalized relation}

We start with a few remarks on unitary operators and one-parameter groups. If $V$ is a unitary operator on $\Hc$, then it is normal, with the spectrum on the unit circle, and with $\Delta^2_\psi(V)=\|(V-\<V\>_\psi)\psi\|^2=1-|\<V\>_\psi|^2\leq 1$. In this case we shall denote
\begin{equation}
 \delta_\psi(V)\equiv\frac{\Delta_\psi(V)}{\big[1-\Delta^2_\psi(V)\big]^{1/2}}=\frac{\big[1-|\<V\>_\psi|^2\big]^{1/2}}{|\<V\>_\psi|}\,.
\end{equation}
This parameter is an increasing function of the deviation, for small spread $\delta_\psi(V)\approx\Delta_\psi(V)$, while for $\Delta_\psi(V)\to1$ (maximal spread) it tends to infinity.

If $V(s)=\exp[-isX]$ is a one-parameter unitary group with the self-adjoint generator $X$, then $\<V(s)\>_\psi=\int_{\sigma(X)} e^{-isx} d\mu_\psi(x)$, where $d\mu_\psi(x)$ is the spectral measure of $X$ in the state $\psi$. Using this representation and its conjugate one finds
\begin{equation}
 \Delta^2_\psi(V(s))=2\int\limits_{\sigma(X)\times\sigma(X)}\!\!\sin^2\big[\tfrac{1}{2}s(x-x')\big]\,d\mu_\psi(x)d\mu_\psi(x')\,.
\end{equation}
Therefore, if $\psi\in\D(X)$, the finite limit exists
\begin{equation}
 \lim_{s\to0}\frac{\Delta^2_\psi(V(s))}{s^2}=\tfrac{1}{2}\int (x-x')^2\,d\mu_\psi(x)d\mu_\psi(x')=\Delta^2_\psi(X)\,,
\end{equation}
which is also equal to $\dsp\lim_{s\to0}s^{-2}\delta^2_\psi(V(s))$ (if $\psi\notin\D(X)$ the limit is $+\infty$).

\subsection{A Weyl pair}

\begin{pr}
Let $U$ and $W$ be unitary operators on $\Hc$, such that
\[
WU=\w\, UW\,,
\]
with some complex number $\w$, $|\w|=1$ -- we shall call any such system a Weyl pair.  Then for each normalized vector $\psi\in\Hc$ the following inequality holds
\begin{equation}\label{WUuncert}
 \tfrac{1}{2}|\w-1|\leq \,\delta_\psi(W)\,\delta_\psi(U)\,.
\end{equation}
\end{pr}
\begin{proof}
We consider relation \eqref{guncert} with the substitutions $A=W$, $B=U$, $\vph=U\psi$, $\chi=W^*\psi$. Then using the relation $UW^*=\w W^*U$ one finds $q_{W,U}(\vph,\chi)=\w-1$, so the lhs of \eqref{guncert} is $|\w-1|$. Also, algebraic relations give $\<W\>_\vph=\w\<W\>_\psi$, $\<W\>_\chi=\<W\>_\psi$, $\<U\>_\vph=\<U\>_\psi$, $\<U\>_\chi=\w\<U\>_\psi$. We set these means into the stationary point condition \eqref{stat}. After some simple algebra this condition is reduced to the form
\begin{equation}
 \Big([\delta_\psi(W)\delta_\psi(U)]^2-[4\la_1\la_2\vep^2]^2\Big)(\la_1-\la_2)=0\,,
\end{equation}
where we have introduced $\vep=\tfrac{1}{2}|\w-1|\in\<0,1\>$.
For the stationary point $\la_1=\la_2=\tfrac{1}{2}$ the relation \eqref{guncert} implies
\begin{equation*}
 \vep\leq\sqrt{\Delta^2_\psi(W)+|\<W\>_\psi|^2\vep^2}\sqrt{\Delta^2_\psi(U)+|\<U\>_\psi|^2\vep^2}\,.
\end{equation*}
Solving this for $\vep$ one obtains relation \eqref{WUuncert}. This stationary point is in fact the unique minimum point of the rhs of our UR \eqref{guncert} in all nontrivial cases, i.e. when $\w\neq1$. Indeed, using inequality \eqref{WUuncert}, and also taking into account that $\la_1\la_2<\tfrac{1}{4}$ for $\la_1\neq\la_2$, we find for $\vep>0$:
\[
 \delta_\psi(W)\delta_\psi(U)-4\la_1\la_2\vep^2>\vep(1-\vep)\geq0\,,
\]
which closes the proof.
\end{proof}
Remark. Relation equivalent to our inequality \eqref{WUuncert} has been given earlier by Massar and Spindel \cite{ma08}, and their proof is to be found in the supplementary material to that reference. Our form of the inequality has more directly visible interpretation. Also, our much simpler proof is an application of our general minimization scheme, showing that the inequality is in a certain sense optimal.

\subsection{The canonical pair}

The standard canonical pair operators $X,P$ are uni\-quely (up to unitary equivalence) defined as being generators of irreducibly represented one-parameter groups $W(\al)=\exp[-i\al X]$, $U(\beta)=\exp[-i\beta P]$ which satisfy relation
\begin{equation}
 W(\al)U(\beta)=\exp[-i\al\beta]U(\beta)W(\al)\,.
\end{equation}
Inequality \eqref{WUuncert} may be written in the form
\begin{equation}
 \frac{\big|\sin\big(\tfrac{1}{2}\al\beta\big)\big|}{|\al\beta|}\leq\frac{\delta_\psi(W(\al))}{|\al|}\frac{\delta_\psi(U(\beta))}{|\beta|}\,.
\end{equation}
For small $\al$, $\beta$ quantities on the rhs approximate the standard deviations of $X$ and $P$, but are finite for all~$\psi$.
If $\psi$ is in the domain of one of the variables, say $\psi\in \D(X)$, then the limit in $\al$ results in
\begin{equation}
 \tfrac{1}{2}\leq\Delta_\psi(X)\frac{\delta_\psi(U(\beta))}{|\beta|}\,,
\end{equation}
and when $\psi\in\D(X)\cap\D(P)$, the usual UR is obtained for $\beta\to0$, with the guarantee of its validity on this domain.

\subsection{Angle -- angular momentum pair}

Similarly, the angle -- angular momentum operators $\Phi$, $L$ are uniquely (up to unitary equivalence) defined by irreducibly represented unitary operators $W(n)=\exp[-in\Phi]$, $n\in\mZ$, and $U(\beta)=\exp[-i\beta L]$, $\beta\in\mR/\modulo 2\pi$, which satisfy relation
\begin{equation}
 W(n)U(\beta)=\exp[-in\beta]U(\beta)W(n)\,.
\end{equation}
 The UR now takes the form
\begin{equation}
 \frac{\big|\sin\big(\tfrac{1}{2}n\beta\big)\big|}{|\beta|}\leq \delta_\psi(W(n))\frac{\delta_\psi(U(\beta))}{|\beta|}\,.
\end{equation}
For $\psi\in\D(L)$ this implies
\begin{equation}
 \tfrac{1}{2}|n|\leq\delta_\psi(W(n))\Delta_\psi(L)\,.
\end{equation}
(The latter relation for $n=1$ was earlier discussed by Hradil \emph{et all} in \cite{hr06}). In particular, if $\psi$ is an eigenvector of $L$, then the spread reaches $\Delta_\psi(W(n))=1$, the maximal value.

\subsection{Unitary transformation}
\begin{pr}
Let $U$ be a unitary transformation, $A$ a normal operator, and denote $A_U=U^*\!A\,U$. Then for each normalized $\chi\in\D(A)\cap\D(A_U)$ $=\D(A)\cap U^*\D(A)$ there is
\begin{equation}\label{utr}
 |\<A_U\>_\chi-\<A\>_\chi|\leq\delta_\chi(U)\big[\Delta_\chi(A_U)+\Delta_\chi(A)\big]\,.
\end{equation}
\end{pr}
\begin{proof}
We set $B=U$ and $\vph=U\chi$ in Proposition \ref{gur}. Then  the lhs of inequality \eqref{guncert} becomes  $|q_{A,U}(U\chi,\chi)|=|\<A_U\>_\chi-\<A\>_\chi|\equiv|\delta\<A\>|$. We also have $\<U\>_\vph=\<U\>_\chi$, $\Delta_\vph(U)=\Delta_\chi(U)$ and $\Delta_\vph(A)=\Delta_\chi(A_U)$. Using these values in the stationary point condition \eqref{stat} we find that its solution is given by $\la_1=\Delta_\chi(A_U)/(\Delta_\chi(A_U)+\Delta_\chi(U))$, $\la_2=1-\la_1$. With these values the rhs of inequality \eqref{guncert} is $\Delta_\chi(U)\big([\Delta_\chi(A_U)+\Delta_\chi(A)]^2+|\delta\<A\>|^2\big)^{1/2}$. Solving now the inequality for $|\delta\<A\>|$ one obtains \eqref{utr}.
\end{proof}

\subsection{Time evolution}

For a quantum system in Heisenberg picture, with the time evolution operator $U(t)=\exp[-itH]$, with $H$ the energy operator, consider Heisenberg normal variable $A_t$. Then $A_{t_2}=(A_{t_1})_{U(t_2-t_1)}$ in the notation of the last subsection. Relation \eqref{utr}, for $\chi\in\D(A_{t_2})\cap\D(A_{t_1})$,  takes now the form
\begin{equation}\label{time}
  |\<A_{t_2}\>_\chi-\<A_{t_1}\>_\chi|\leq\delta_\chi(U(t_2-t_1))\big[\Delta_\chi(A_{t_2})+\Delta_\chi(A_{t_1})\big]\,.
\end{equation}
If $\dsp\chi\in\bigcap_{\tau\in(t-\vep,t+\vep)}\D(A_\tau)\cap\D(H)$ and such that $A_t\chi$ is norm-continuous, one can recover the well-known relation
\begin{equation}\label{dt}
  \tfrac{1}{2}\Big|\frac{d\<A_t\>_\chi}{dt}\Big|\leq\Delta_\chi(H)\Delta_\chi(A_t)\,.
\end{equation}
It is well-known that there is no self-adjoint time operator which would be translated by time evolution as $U(t)^*TU(t)=T+t$ (see, e.g., \cite{bu95}).\footnote{We insist on self-adjointness; non-self-adjoint `time operators' considered in literature (a recent example is \cite{hi09}) do not have spectral decompositions, thus the probabilistic interpretation does not apply in standard form.} Formula \eqref{dt} may be interpreted as a substitute for time -- energy relation, in the following sense \cite{ma45}. Let $d\<A_t\>_\chi/dt\approx\con.$ in some interval of $t$, then $A$ may be rescaled so that in fact $d\<A_t\>_\chi/dt\approx1$. Then $\<A_t\>_\chi$ describes correctly the flow of time in this interval. The product of uncertainties is then bounded by $1/2$ from below.

The full formula \eqref{time} tells us more. Suppose  $\<A_t\>_\chi$ is an increasing function of time, say, for simplicity, proportional to $t$. For this quantity to be a good measure of time we demand that the deviation $\Delta_\chi(A_t)$ stays bounded by a constant. This is then possible only if $|\<U(t)\>_\chi|=|(\chi,\chi_t)|$ decreases at least as $1/t$ with time -- the Schr\"odinger evolution of the state has to bring it sufficiently fast away from the initial state.\footnote{On other state decay estimates -- in contex of non-self-adjoint time operators -- see \cite{ar05}.}

\subsection{Angular momentum}

Let $\Hc$ be a representation space of the usual angular momentum operators~$J_i$, either bosonic: \mbox{$\exp(i2\pi J_3)=\1$}, or fermionic: $\exp(i2\pi J_3)=- \1$. For \mbox{$\vph,\chi\in\D(\sqrt{J^2})$} the weak commutator of $J_i$ and $J_j$ is equal to the strong one, that is $q_{J_1,J_2}(\vph,\chi)=i(\vph,J_3\chi)$ (and permuted relations). This is easily seen by differentiating the identity
\begin{equation}
 (e^{-i\alpha J_1}\vph,e^{i\beta J_2}\chi)-(e^{-i\beta J_2}\vph,e^{i\alpha J_1}\chi)
 =(e^{-i\alpha J_1}\vph,[e^{i\beta J_2}-e^{i\beta(\cos\alpha J_2+\sin\alpha J_3)}]\chi)
\end{equation}
with respect to $\beta$, setting $\beta=0$, and then performing the same operation with respect to $\alpha$. Therefore for such vectors and any real numbers $j_1,j_2$ we obtain
\begin{equation}\label{angmom}
 |(\vph,J_3\psi)| \leq\|(J_1-j_1)\vph\|\,\|(J_2-j_2)\psi\|
 + \|(J_2-j_2)\vph\|\,\|(J_1-j_1)\psi\|\,.
\end{equation}
There are now a few possible choices of $\vph$. The standard choice $\vph=\psi$ and minimalization with respect to $j_1,j_2$ gives the standard relation \begin{equation}
 \tfrac{1}{2}|\<J_3\>_\psi|\leq\Delta_\psi(J_1)\Delta_\psi(J_2)\,,\quad \psi\in\D(\sqrt{J^2})
\end{equation}
on maximal possible domain. A~weak point of this relation is that $J_3$ is not positive, so the mean on the lhs may take arbitrarily small value for certain states, including zero.

If the spectrum of $J^2$ is bounded, say $\|J^2\|=j(j+1)$, then $\|J_i\|=j$, $|\<J_i\>_\psi|\leq j$, so setting $j_i=\<J_i\>_\psi$ and taking supremum over $\vph$, $\|\vph\|=1$, we find
\[
 \|J_3\psi\|\leq 2j\big(\Delta_\psi(J_1)+\Delta_\psi(J_2)\big)\,.
\]
The lhs now vanishes only for $J_3\psi=0$, in which case $\Delta_\psi^2(J_1)=\Delta_\psi^2(J_2)=\tfrac{1}{2}\<J^2\>_\psi$. For $J_i$ being spin $1/2$ operators the relation is $\tfrac{1}{2}\leq\Delta_\psi(J_1)+\Delta_\psi(J_2)$.

Finally, we consider the general case of unbounded $J^2$. We denote by $P_m$ the projection operator onto the eigensubspace $\Ker(J_3-m\1)$. Let $\mu$ be a spectral value of $J_3$ such that $\mu-1\leq0\leq\mu$ (thus in fermionic case $\mu=1/2$, while in bosonic case $\mu=0$ or $\mu=1$). We introduce further self-adjoint operators
\begin{equation}
 P=P_{\mu}+ P_{\mu-1}\,,\quad E=\Big(\sum_{m\geq \mu}-\sum_{m\leq\mu-1}\Big)P_m
\end{equation}
so that $E^2=\1$, $P^2=P$, $J_3=E|J_3|$.
\begin{pr}
In standing notation, for $\psi\in\D(\sqrt{J^2})$, there is
\begin{equation}
 (\psi,|J_3|\psi)\leq 2\Delta_\psi(J_1)\Delta_\psi(J_2)
 +\big\|\big[J^2+\tfrac{1}{4}\delta\big]^{1/2}P\psi\big\|\,
 \Big(\Delta_\psi(J_1)+\Delta_\psi(J_2)\Big)\,,
\end{equation}
where $\delta=0$ ($\delta=1$) in bosonic (fermionic) case, respectively.
\end{pr}
\begin{proof}
Expressing operators $J_i$ ($i=1,2$) in terms of $J_\pm$ one shows that $EJ_iE-J_i=-2PJ_iP\equiv W_i$, and furthermore, $W_i^2=\big(J^2+\mu(1-\mu)\big)P=(J^2+\tfrac{1}{4}\delta)P$, with $\delta$ defined in the thesis. It follows that $|W_i|=[J^2+\tfrac{1}{4}\delta]^{1/2}P$. Setting now $\vph=E\psi$, $j_i=\<J_i\>_\psi$ we note that
\[
\|(J_i-j_i)E\psi\|=\|(EJ_iE-j_i)\psi\|\leq \|(J_i-j_i)\psi\|+ \|W_i\psi\|\,,
\]
which gives the thesis when used in \eqref{angmom}.
\end{proof}

In particular, if $P_0\psi=P_{+1}\psi=0$ or $P_0\psi=P_{-1}\psi=0$ in bosonic case, or $P_{-1/2}\psi=P_{+1/2}\psi=0$ in fermionic case, then
\[
 \tfrac{1}{2}|\<|J_3|\>_\psi|\leq\Delta_\psi(J_1)\Delta_\psi(J_2)\,.
\]

\end{document}